%% file: sigir_paper.tex
\documentclass[sigconf]{acmart}

\usepackage{todonotes}
\usepackage{amssymb}
\usepackage{amsmath}
\usepackage{bm}
\usepackage[ruled,vlined]{algorithm2e}
\usepackage{caption}
\usepackage{subcaption}

\input{macros.tex}

\AtBeginDocument{%
  \providecommand\BibTeX{{%
    \normalfont B\kern-0.5em{\scshape i\kern-0.25em b}\kern-0.8em\TeX}}}

\setcopyright{acmcopyright}
\copyrightyear{2020}
\acmYear{2020}
\acmDOI{10.1145/1122445.1122456}

\acmConference[ACM SIGIR]{ACM SIGIR Conference on Research and Development in Information Retrieval}{2020}{Xi'an, China}
\acmBooktitle{ACM SIGIR Conference on Research and Development in Information Retrieval}
\acmPrice{15.00}
\acmISBN{978-1-4503-9999-9/18/06}

\begin{document}
\fancyhead{}
\title{Controlling Fairness and Bias in Dynamic Learning-to-Rank}

\author{Marco Morik}
\authornote{Equal contribution.}
\authornote{Work conducted while at Cornell University.}
\email{m.morik@tu-berlin.de}
\affiliation{
    \institution{Technische Univerit\"at Berlin}
    \streetaddress{Machine Learning Group}
    \city{Berlin}
    \country{Germany}
    \postcode{10587}
}

\author{Ashudeep Singh}
\authornotemark[1]
\email{ashudeep@cs.cornell.edu}
\affiliation{%
  \institution{Cornell University}
  \streetaddress{Department of Computer Science}
  \city{Ithaca}
  \state{NY}
  \postcode{14853}
}

\author{Jessica Hong}
\email{jwh296@cornell.edu}
\affiliation{%
  \institution{Cornell University}
  \streetaddress{Department of Computer Science}
  \city{Ithaca}
  \state{NY}
  \postcode{14853}
}

\author{Thorsten Joachims}
\email{tj@cs.cornell.edu}
\affiliation{%
  \institution{Cornell University}
  \streetaddress{Department of Computer Science}
  \city{Ithaca}
  \state{NY}
  \postcode{14853}
}

\renewcommand{\shortauthors}{Morik et al.}

\begin{abstract}
Rankings are the primary interface through which many online platforms match users to items (e.g. news, products, music, video). In these two-sided markets, not only the users draw utility from the rankings, but the rankings also determine the utility (e.g. exposure, revenue) for the item providers (e.g. publishers, sellers, artists, studios). It has already been noted that myopically optimizing utility to the users -- as done by virtually all learning-to-rank algorithms -- can be unfair to the item providers. We, therefore, present a learning-to-rank approach for explicitly enforcing merit-based fairness guarantees to groups of items (e.g. articles by the same publisher, tracks by the same artist).
In particular, we propose a learning algorithm that ensures notions of amortized group fairness, while simultaneously learning the ranking function from implicit feedback data. The algorithm takes the form of a controller that integrates unbiased estimators for both fairness and utility, dynamically adapting both as more data becomes available. In addition to its rigorous theoretical foundation and convergence guarantees, we find empirically that the algorithm is highly practical and robust. 

\end{abstract}




\copyrightyear{2020}
\acmYear{2020}
\setcopyright{acmlicensed}\acmConference[SIGIR '20]{Proceedings of the 43rd International ACM SIGIR Conference on Research and Development in Information Retrieval}{July 25--30, 2020}{Virtual Event, China}
\acmBooktitle{Proceedings of the 43rd International ACM SIGIR Conference on Research and Development in Information Retrieval (SIGIR '20), July 25--30, 2020, Virtual Event, China}
\acmPrice{15.00}
\acmDOI{10.1145/3397271.3401100}
\acmISBN{978-1-4503-8016-4/20/07}

\maketitle

\section{Introduction}

We consider the problem of dynamic Learning-to-Rank (LTR), where the ranking function dynamically adapts based on the feedback that users provide. Such dynamic LTR problems are ubiquitous in online systems --- news-feed rankings that adapt to the number of "likes" an article receives, online stores that adapt to the number of positive reviews for a product, or movie-recommendation systems that adapt to who has watched a movie. In all of these systems, learning and prediction are dynamically intertwined, where past feedback influences future rankings in a specific form of online learning with partial-information feedback \cite{CesaBianchi/Lugosi/06}. 

While dynamic LTR systems are in widespread use and unquestionably useful, there are at least two issues that require careful design considerations. First, the ranking system induces a bias through the rankings it presents. In particular, items ranked highly are more likely to collect additional feedback, which in turn can influence future rankings and promote misleading rich-get-richer dynamics \cite{adamic2000power,Salganik06,Joachims07a,joachims17unbiased}. Second, the ranking system is the arbiter of how much exposure each item receives, where exposure directly influences opinion (e.g. ideological orientation of presented news articles) or economic gain (e.g. revenue from product sales or streaming) for the provider of the item. This raises fairness considerations about how exposure should be allocated based on the merit of the items \cite{singh2018fairness,biega2018equity}. We will show in the following that naive dynamic LTR methods that are oblivious to these issues can lead to economic disparity, unfairness, and polarization. 

In this paper, we present the first dynamic LTR algorithm -- called \conname\ -- that overcomes rich-get-richer dynamics while enforcing a configurable allocation-of-exposure scheme. Unlike existing fair LTR algorithms \cite{singh2018fairness,biega2018equity,singh2019policy,yadav2019fair}, \conname\ explicitly addresses the dynamic nature of the learning problem, where the system is unbiased and fair even though the relevance and the merit of items are still being learned. At the core of our approach lies a merit-based exposure-allocation criterion that is amortized over the learning process \cite{singh2018fairness,biega2018equity}. We view the enforcement of this merit-based exposure criterion as a control problem and derive a P-controller that optimizes both the fairness of exposure as well as the quality of the rankings. A crucial component of the controller is the ability to estimate merit (i.e. relevance) accurately, even though the feedback is only revealed incrementally as the system operates, and the feedback is biased by the rankings shown in the process \cite{Joachims07a}. To this effect, \conname\ includes a new unbiased cardinal relevance estimator -- as opposed to existing ordinal methods \cite{joachims17unbiased,agarwal19general} --, which can be used both as an unbiased merit estimator for fairness and as a ranking criterion.  

In addition to the theoretical justification of \conname, we provide empirical results on both synthetic news-feed data and real-world movie recommendation data. We find that \conname\ is effective at enforcing fairness while providing good ranking performance. Furthermore, \conname\ is efficient, robust, and easy to implement.

\section{Motivation} \label{sec:motivation}

Consider the following illustrative example of a dynamic LTR problem. An online news-aggregation platform wants to present a ranking of the top news articles on its front page. Through some external mechanism, it identifies a set $\Y=\{\ypart_1,...,\ypart_{20}\}$ of 20 articles at the beginning of each day, but it is left with the learning problem of how to rank these 20 articles on its front page. As users start coming to the platform, the platform uses the following naive algorithm to learn the ranking.

\begin{algorithm}%
Initialize counters $C(\ypart)=0$ for each $\ypart \in \Y$\;
\ForEach{user}{
present ranking $\y=\argsort_\Y[C(\ypart)]$ (random tiebreak)\; 
increment $C(\ypart)$ for the articles read by the user.
}
\caption{Naive Dynamic LTR Algorithm}
\end{algorithm}

Executing this algorithm at the beginning of a day, the platform starts by presenting the 20 articles in random order for the first user. It may then observe that the user reads the article in position 3 and increments the counter $C(\ypart)$ for this article. For the next user, this article now gets ranked first and the counters are updated based on what the second user reads. This cycle continues for each subsequent user. Unfortunately, this naive algorithm has at least two deficiencies that make it suboptimal or unsuitable for many ranking applications. 

The first deficiency lies in the choice of $C(\ypart)$ as an estimate of average relevance for each article -- namely the fraction of users that want to read the article. Unfortunately, even with infinite amounts of user feedback, the counters $C(\ypart)$ are not consistent estimators of average relevance \cite{Salganik06,Joachims07a,joachims17unbiased}. In particular, items that happened to get more reads in early iterations get ranked highly, where more users find them and thus have the opportunity to provide more positive feedback for them. This perpetuates a rich-get-richer dynamic, where the feedback count $C(\ypart)$ recorded for each article does not reflect how many users actually wanted to read the article. 

The second deficiency of the naive algorithm lies in the ranking policy itself, creating a source of unfairness even if the true average relevance of each article was accurately known \cite{singh2018fairness,biega2018equity,trec2019fair}. Consider the following omniscient variant of the naive algorithm that ranks the articles by their true average relevance (i.e. the true fraction of users who want to read each article). How can this ranking be unfair? Let us assume that we have two groups of articles, $G_{\rght}$ and $G_{\lft}$, with 10 items each (i.e. articles from politically right- and left-leaning sources). 51\% of the users (right-leaning) want to read the articles in group $G_{\rght}$, but not the articles in group $G_{\lft}$. In reverse, the remaining 49\% of the users (left-leaning) like only the articles in $G_{\lft}$. Ranking articles solely by their true average relevance puts items from $G_{\rght}$ into positions 1-10 and the items from $G_{\lft}$ in positions 11-20. This means the platform gives the articles in $G_{\lft}$ vastly less exposure than those in $G_{\rght}$. We argue that this can be considered unfair since the two groups receive disproportionately different outcomes despite having similar merit (i.e. relevance). Here, a 2\% difference in average relevance leads to a much larger difference in exposure between the groups.  

We argue that these two deficiencies -- namely bias and unfairness -- are not just undesirable in themselves, but that they have undesirable consequences. For example, biased estimates lead to poor ranking quality, and unfairness is likely to alienate the left-leaning users in our example, driving them off the platform and encouraging polarization.

Furthermore, note that these two deficiencies are not specific to the news example, but that the naive algorithm leads to analogous problems in many other domains. For example, consider a ranking system for job applicants, where rich-get-richer dynamics and exposure allocation may perpetuate and even amplify existing unfairness (e.g. disparity between male and female applicants). Similarly, consider an online marketplace where products of different sellers (i.e. groups) are ranked. Here rich-get-richer dynamics and unfair exposure allocation can encourage monopolies and drive some sellers out of the market.  

These examples illustrate the following two desiderata that a less naive dynamic LTR algorithm should fulfill.
\begin{description}
    \item[Unbiasedness:] The algorithm should not be biased or subject to rich-get-richer dynamics.
    \item[Fairness:] The algorithm should enforce a fair allocation of exposure based on merit (e.g. relevance).
\end{description}

With these two desiderata in mind, this paper develops alternatives to the Naive algorithm. In particular, after introducing the dynamic learning-to-rank setting in Section~\ref{sec:dynltr}, Section~\ref{sec:fairness} formalizes an amortized notion of merit-based fairness, accounting for the fact that merit itself is unknown at the beginning of the learning process and is only learned throughout. Section~\ref{sec:relest} then addresses the bias problem, providing estimators that eliminate the presentation bias for both global and personalized ranking policies. Finally, Section~\ref{sec:faircontrol} proposes a control-based algorithm that is designed to optimize ranking quality while dynamically enforcing fairness.

\section{Related Work} \label{sec:related}

Ranking algorithms are widely recognized for their potential for societal impact \cite{baeza2018bias}, as they form the core of many online systems, including search engines, recommendation systems, news feeds, and online voting. Controlling rich-get-richer phenomena in recommendations and rankings has been studied from the perspective of both optimizing utility through exploration as well as ensuring fairness of such systems \cite{yin2012challenging, schnabel2016recommendations, abdollahpouri2017controlling}. There are several adverse consequences of naive ranking systems \cite{ciampaglia2018algorithmic}, such as political polarization \cite{beam2014automating}, misinformation \cite{vosoughi2018spread}, unfair allocation of exposure \cite{singh2019policy}, and biased judgment \cite{baeza2018bias} through phenomena such as the Matthew effect \cite{adamic2000power,germano2019few}. Viewing such ranking problems as two-sided markets of users and items that each derive utility from the ranking system brings a novel perspective to tackling such problems \cite{singh2018fairness, abdollahpouri2019beyond}. In this work, we take inspiration from these works to develop methods for mitigating bias and unfairness in a dynamic setting.

Machine learning methods underlie most ranking algorithms. There has been a growing concern around the question of how machine learning algorithms can be unfair, especially given their numerous real-world applications \cite{barocas2016big}. There have been several definitions proposed for fairness in the binary classification setting \cite{barocasfairness}, as well as recently in the domains of rankings in recommendations and information retrieval \cite{celis2017ranking,singh2018fairness,biega2018equity,beutel2019fairness}. The definitions of fairness in ranking span from ones purely based on the composition of the top-k \cite{celis2017ranking}, to relevance-based definitions such as fairness of exposure \cite{singh2018fairness}, and amortized attention equity \cite{biega2018equity}. We will discuss these definitions in greater detail in Section~\ref{sec:fairness}. Our work also relates to the recent interest in studying the impact of fairness when learning algorithms are applied in dynamic settings \cite{liu2018delayed,ensign2018runaway,tabibian2019consequential}. 

In information retrieval, there has been a long-standing interest in learning to rank from biased click data. As already argued above, the bias in logged click data occurs because the feedback is incomplete and biased by the presentation. Numerous approaches based on preferences (e.g. \cite{Herbrich:1202,Joachims/02c}), click models (e.g. \cite{chuklin2015click}), and randomized interventions (e.g. \cite{radlinski06minimally}) exist. Most recently, a new approach for de-biasing feedback data using techniques from causal inference and missing data analysis was proposed to provably eliminate selection biases \cite{joachims17unbiased, ai2018unbiased}. We follow this approach in this paper, extend it to the dynamic ranking setting, and propose a new unbiased regression objective in Section~\ref{sec:relest}.

Learning in our dynamic ranking setting is related to the conventional learning-to-rank algorithms such as LambdaRank, LambdaMART, RankNet, Softrank etc. \cite{burges2010ranknet,taylor2008softrank}. However, to implement fairness constraints based on merit, we need to explicitly estimate relevance to the user as a measure of merit while the scores estimated by these methods don't necessarily have a meaning. Our setting is also closely related to online learning to rank for top-k ranking where feedback is observed only on the top-k items, and hence exploration interventions are necessary to ensure convergence \cite{Radlinski08learning, hofmann2013balancing, zoghi2017online, li2018online}. These algorithms are designed with respect to a click-model assumption \cite{zoghi2017online} or learning in the presence of document features \cite{li2018online}. A key difference in our method is that we do not consider exploration through explicit interventions, but merely exploit user-driven exploration. However, explicit exploration could also be incorporated into our algorithms to improve the convergence rate of our methods.

\section{Dynamic Learning-to-Rank} \label{sec:dynltr}

We begin by formally defining the dynamic LTR problem. Given is a set of items $\Y$ that needs to be ranked in response to incoming requests. At each time step $t$, a request 
\begin{equation}
    \x_t,\relx_t \sim \Prob(\x,\relx)
\end{equation}
arrives i.i.d.\ at the ranking system. Each request consists of a feature vector describing the user's information need $\x_t$ (e.g. query, user profile), and the user's vector of true relevance ratings $\relx_t$ for all items in the collection $\Y$. Only the feature vector $\x_t$ is visible to the system, while the true relevance ratings $\relx_t$ are hidden. Based on the information in $\x_t$, a ranking policy $\pol_t(\x)$ produces a ranking $\y_t$ that is presented to the user. Note that the policy may ignore the information in $\x_t$, if we want to learn a single global ranking like in the introductory news example. 

After presenting the ranking $\y_t$, the system receives a feedback vector $\click_t$ from the user with a non-negative value $\click_t(\ypart)$ for every $\ypart \in \Y$. In the simplest case, it is $1$ for click and $0$ for no click, and we will use the word "click" as a placeholder throughout this paper for simplicity. But the feedback may take many other forms and does not have to be binary. For example, in a video streaming service, the feedback may be the percentage the user watched of each video.

After the feedback $\click_t$ was received, the dynamic LTR algorithm $\dynrank$ now updates the ranking policy and produces the policy $\pol_{t+1}$ that is used in the next time step. 
\[
    \pol_{t+1} \longleftarrow \dynrank((\x_1,\y_1,\click_1),...,(\x_t,\y_t,\click_t)) 
\]
An instance of such a dynamic LTR algorithm is the Naive algorithm already outlined in Section~\ref{sec:motivation}. It merely computes $\sum \click_t$ to produce a new ranking policy for $t+1$ (here a global ranking independent of $\x$). 

\subsection{Partial and Biased Feedback}

A key challenge of dynamic LTR lies in the fact that the feedback $\click_t$ provides meaningful feedback only for the items that the user examined. Following a large body of work on click models \cite{chuklin2015click}, we model this as a censoring process. Specifically, for a binary vector $\expx_t$ indicating which items were examined by the user, we model the relationship between $\click_t$ and $\relx_t$ as follows.
\begin{equation}
    \click_t(\ypart) = 
    \left\{
        \begin{matrix}
        \relx_t(\ypart) & \mbox{if $\expx_t(\ypart)=1$}\\
        0               & \mbox{otherwise} 
        \end{matrix}
    \right. 
\end{equation}
Coming back to the running example of news ranking, $\relx_t$ contains the full information about which articles the user is interested in reading, while $\click_t$ reveals this information only for the articles $\ypart$ examined by the user (i.e. $\expx_t(\ypart)=1$). Analogously, in the job placement application $\relx_t$ indicates for all candidates $\ypart$ whether they are qualified to receive an interview call, but $\click_t$ reveals this information only for those candidates examined by the employer. 

A second challenge lies in the fact that the examination vector $\expx_t$ cannot be observed. This implies that a feedback value of $\click_t(\ypart)=0$ is ambiguous -- it may either indicate lack of examination (i.e. $\expx_t(\ypart)=0$) or negative feedback (i.e. $\relx_t(\ypart)=0$). This would not be problematic if $\expx_t$ was uniformly random, but which items get examined is strongly biased by the ranking $\y_t$ presented to the user in the current iteration. Specifically, users are more likely to look at an item high in the ranking than at one that is lower down \cite{Joachims07a}. We model this position bias as a probability distribution on the examination vector
\begin{equation}
    \expx_t \sim \Prob(\expx | \y_t,\x_t,\relx_t). 
\end{equation}
Most click models can be brought into this form \cite{chuklin2015click}. For the simplicity of this paper, we merely use the Position-Based Model (PBM) \cite{craswell2008experimental}. It assumes that the marginal probability of examination $\prop_t(\ypart)$ for each item $\ypart$ depends only on the rank $\rank(\ypart|\y)$ of $\ypart$ in the presented ranking $\y$. Despite its simplicity, it was found that the PBM can capture the main effect of position bias accurately enough to be reliable in practice \cite{joachims17unbiased,wang2018position,agarwal19estimating}. 

\subsection{Evaluating Ranking Performance}

We measure the quality of a ranking policy $\pol$ by its utility to the users. Virtually all ranking metrics used in information retrieval define the utility $\Loss(\y|\relx)$ of a ranking $\y$ as a function of the relevances of the individual items $\relx$. In our case, these item-based relevances $\relx$ represent which articles the user likes to read, or which candidates are qualified for an interview. A commonly used utility measure is the DCG \cite{jarvelin2002cumulated}
\[
    \Loss^{DCG}(\y|\relx) =  \sum_{\ypart \in \y}  \frac{\relx(\ypart)}{\log_2(1+\rank(\ypart|\y))},
\]
or the NDCG when normalized by the DCG of the optimal ranking. Over a distribution of requests $\Prob(\x,\relx)$, a ranking policy $\pol(\x)$ is evaluated by its expected utility
\begin{eqnarray}
    \Util(\pol) & = & \int \Loss(\pol(\x)|\relx) \: d \Prob(\x,\relx). \label{eq:utilfullinfo}
\end{eqnarray}

\subsection{Optimizing Ranking Performance} \label{sec:sortrank}

The user-facing goal of dynamic LTR is to converge to the policy $\pol^*=\argmax_\pol \Util(\pol)$ that maximizes utility. Even if we solve the problem of estimating $\Util(\pol)$ despite our lack of knowledge of $\expx$, this maximization problem could be computationally challenging, since the space of ranking policies is exponential even when learning just a single global ranking. Fortunately, it is easy to show \cite{robertson1977probability} that sorting-based policies 
\begin{eqnarray}
    \pol(\x)  \equiv  \argsort_{\ypart \in \Y} \big[ \Relevance(\ypart|\x) \big],  \label{eq:sortrank}
\end{eqnarray}
where
\begin{eqnarray}
     \Relevance(\ypart|\x)=\int \relx(\ypart) \: d \Prob(\relx|\x),   \label{eq:condrel}
\end{eqnarray}
are optimal for virtually all $\Loss(\y|\relx)$ commonly used in IR (e.g. DCG). So, the problem lies in estimating the expected relevance $\Relevance(\ypart|\x)$ of each item $\ypart$ conditioned on $\x$. When learning a single global ranking, this further simplifies to estimating the expected average relevance  $\Relevance(\ypart) = \int \relx(\ypart) \: d \Prob(\relx,\x)$ for each item $\ypart$. The global ranking can then be derived via
\begin{eqnarray}
    \y  = \argsort_{\ypart \in \Y} \big[ \Relevance(\ypart) \big]   \label{eq:sortrankglob}
\end{eqnarray}
In Section~\ref{sec:relest}, we will use techniques from causal inference and missing-data analysis to design unbiased and consistent estimators for $\Relevance(\ypart|\x)$ and $\Relevance(\ypart)$ that only require access to the observed feedback $\click_t$.

\section{Fairness in Dynamic LTR} \label{sec:fairness}

While sorting by $\Relevance(\ypart|\x)$ (or $\Relevance(\ypart)$ for global rankings) may provide optimal utility to the user, the introductory example has already illustrated that this ranking can be unfair. There is a growing body of literature to address this unfairness in ranking, and we now extend merit-based fairness \cite{singh2018fairness,biega2018equity} to the dynamic LTR setting.

The key scarce resource that a ranking policy allocates among the items is exposure. Based on the model introduced in the previous section, we define the exposure of an item $\ypart$ as the marginal probability of examination $\prop_t(\ypart)=\Prob(\expx_t(\ypart)=1 | \y_t,\x_t,\relx_t)$. It is the probability that the user will see $\ypart$ and thus have the opportunity to read that article, buy that product, or interview that candidate. We discuss in Section~\ref{sec:relest} how to estimate $\prop_t(\ypart)$.
Taking a group-based approach to fairness, we aggregate exposure by groups $\mathcal{G} = \{G_1,\dots,G_m\}$.
\begin{eqnarray}
    \Exposure_t(G_i) & = & \frac{1}{|G_i|}\sum_{\ypart \in G_i} \prop_t(\ypart). \label{eq:expg}
\end{eqnarray}
These groups can be legally protected groups (e.g. gender, race), reflect some other structure (e.g. items sold by a particular seller), or simply put each item in its own group (i.e. individual fairness). 

In order to formulate fairness criteria that relate exposure to merit, we define the merit of an item as its expected average relevance $\Relevance(\ypart)$ and again aggregate over groups.
\begin{eqnarray}
    \Merit(G_i) = \frac{1}{|G_i|}\sum_{\ypart \in G_i} \Relevance(\ypart) \label{eq:merit}
\end{eqnarray}
In Section~\ref{sec:relest}, we will discuss how to get unbiased estimates of $\Merit(G_i)$ using the biased feedback data $\click_t$. 

With these definitions in hand, we can address the types of disparities identified in Section~\ref{sec:motivation}. Specifically, we extend the Disparity of Treatment criterion of \cite{singh2018fairness} to the dynamic ranking problem, using an amortized notion of fairness as in \cite{biega2018equity}. In particular, for any two groups $G_i$ and $G_j$ the disparity
\begin{eqnarray}
    \DiE_{\tau}(G_i,G_j) = \frac{\frac{1}{\tau}\sum_{t=1}^\tau \Exposure_t(G_i)}{ \Merit(G_i) } - \frac{\frac{1}{\tau}\sum_{t=1}^\tau \Exposure_t(G_j)}{\Merit(G_j)} \label{eq:fairexposure}
\end{eqnarray}
measures in how far amortized exposure over $\tau$ time steps was fulfilled.
This {\bf exposure-based fairness disparity} expresses in how far, averaged over all time steps, each group of items got exposure proportional to its relevance. The further the disparity is from zero, the greater is the violation of fairness. Note that other allocation strategies beyond proportionality could be implemented as well by using alternate definitions of disparity \cite{singh2018fairness}. 

Exposure can also be allocated based on other fairness criteria, for example, a Disparity of Impact that a specific exposure allocation implies \cite{singh2018fairness}. If we consider the feedback $\click_t$ (e.g. clicks, purchases, votes) as a measure of impact
\begin{eqnarray}
    \Impact_t(G_i) & = & \frac{1}{|G_i|}\sum_{\ypart \in G_i} \click_t(\ypart), \label{eq:impg}
\end{eqnarray}
then keeping the following disparity close to zero controls how exposure is allocated to make impact proportional to relevance.
\begin{eqnarray}
    \DiI_{\tau}(G_i,G_j) = \frac{\frac{1}{\tau}\sum_{t=1}^\tau \Impact_t(G_i)}{ \Merit(G_i) } - \frac{\frac{1}{\tau}\sum_{t=1}^\tau \Impact_t(G_j)}{\Merit(G_j) } \label{eq:fairimpact}
\end{eqnarray}
We refer to this as the {\bf impact-based fairness disparity}. In Section~\ref{sec:faircontrol} we will derive a controller that drives such exposure and impact disparities to zero.

\section{Unbiased Estimators} \label{sec:relest}

To be able to implement the ranking policies in Equation~\eqref{eq:sortrank} and the fairness disparities in Equations~\eqref{eq:fairexposure} and \eqref{eq:fairimpact}, we need accurate estimates of the position bias $\prop_t$, the expected conditional relevances $\Relevance(\ypart|\x)$, and the expected average relevances $\Relevance(\ypart)$. We consider these estimation problems in the following.

\subsection{Estimating the Position Bias}

Learning a model for $\prop_t$ is not part of our dynamic LTR problem, as the position-bias model is merely an input to our dynamic LTR algorithms. Fortunately, several techniques for estimating position-bias models already exist in the literature \cite{joachims17unbiased,wang2018position,agarwal19estimating,Fang19a}, and we are agnostic to which of these is used. In the simplest case, the examination probabilities $\prop_t(\ypart)$ only depend on the rank of the item in $\y$, analogous to a Position-Based Click Model \cite{craswell2008experimental} with a fixed probability for each rank. It was shown in \cite{joachims17unbiased,wang2018position,agarwal19estimating} how these position-based probabilities can be estimated from explicit and implicit swap interventions. Furthermore, it was shown in \cite{Fang19a} how the contextual features $\x$ about the users and query can be incorporated in a neural-network based propensity model, allowing it to capture that certain users may explore further down the ranking for some queries. Once any of these propensity models are learned, they can be applied to predict $\prop_t$ for any new query $\x_t$ and ranking $\y_t$.

\subsection{Estimating Conditional Relevances} \label{sec:ipscondutil}

The key challenge in estimating $\Relevance(\ypart|\x)$ from Equation~\eqref{eq:condrel} lies in our inability to directly observe the true relevances $\relx_t$. Instead, the only data we have is the partial and biased feedback $\click_t$. To overcome this problem, we take an approach inspired by \cite{joachims17unbiased} and extend it to the dynamic ranking setting. The key idea is to correct for the selection bias with which relevance labels are observed in $\click_t$ using techniques from survey sampling and causal inference \cite{horvitz1952generalization,imbens2015causal}. However, unlike the ordinal estimators proposed in \cite{joachims17unbiased}, we need cardinal relevance estimates since our fairness disparities are cardinal in nature. We, therefore, propose the following cardinal relevance estimator.

The key idea behind this estimator lies in a training objective that only uses $\click_t$, but that in expectation is equivalent to a least-squares objective that has access to $\relx_t$. To start the derivation, let's consider how we would estimate $\Relevance(\ypart|\x)$, if we had access to the relevance labels $(\relx_1, ...,\relx_\tau)$ of the previous $\tau$ time steps. A straightforward solution would be to solve the following least-squares objective for a given regression model $\RelevanceReg{\w}(\ypart|\x_t)$ (e.g. a neural network), where $\w$ are the parameters of the model. 
\begin{eqnarray}
    \L^{\relx}(\w) & = & \sum_{t=1}^\tau \sum_{\ypart} \left(\relx_t(\ypart) - \RelevanceReg{\w}(\ypart|\x_t)\right)^2  \label{eq:regrel}
\end{eqnarray}
The minimum $\w^*$ of this objective is the least-squares regression estimator of $\Relevance(\ypart|\x_t)$. Since the $(\relx_1, ...,\relx_\tau)$ are not available, we define an asymptotically equivalent objective that merely uses the biased feedback $(\click_1, ...,\click_\tau)$. The new objective corrects for the position bias using Inverse Propensity Score (IPS) weighting \cite{horvitz1952generalization,imbens2015causal}, where the position bias $(\prop_1,...,\prop_\tau)$ takes the role of the missingness model. 
\begin{eqnarray}
    \L^{\click}(\w) \;=\; \sum_{t=1}^\tau \sum_{\ypart} \RelevanceReg{\w}(\ypart|\x_t)^2 + \frac{\click_t(\ypart)}{\prop_t(\ypart)}  (\click_t(\ypart)-2 \RelevanceReg{\w}(\ypart|\x_t)) \label{eq:newloss2}
\end{eqnarray}
We denote the regression estimator defined by the minimum of this objective as $\RelevanceRegStar(\ypart|\x_t)$.
The regression objective in \eqref{eq:newloss2} is unbiased, meaning that its expectation is equal to the regression objective $\L^{\relx}(\w)$ that uses the unobserved true relevances $(\relx_1, ...,\relx_\tau)$.
\begin{align*}
    \Exp_{\expx}&\left[\L^{\click}(\w)\right]  \\
    &= \sum_{t=1}^\tau \!\sum_{\ypart} \!\!\sum_{\expx_t(\ypart)} \!\!\left[\RelevanceReg{\w}(\ypart|\x_t)^2 \!+\! \frac{\click_t(\ypart)}{\prop_t(\ypart)}  (\click_t(\ypart)\!-\!2 \RelevanceReg{\w}(\ypart|\x_t\!))\right] \!\Prob(\expx_t(\ypart)|\y_t,\x_t\!) \\
    &= \sum_{t=1}^\tau \sum_{\ypart} \RelevanceReg{\w}(\ypart|\x_t)^2 + \frac{1}{\prop_t(\ypart)}  \relx_t(\ypart)  (\relx_t(\ypart) -2 \RelevanceReg{\w}(\ypart|\x_t)) \prop_t(\ypart) \\
    &= \sum_{t=1}^\tau \sum_{\ypart}  \RelevanceReg{\w}(\ypart|\x_t)^2 + \relx_t(\ypart)^2 - 2 \relx_t(\ypart) \RelevanceReg{\w}(\ypart|\x_t) \\
    &= \sum_{t=1}^\tau \sum_{\ypart}  \left(\relx_t(\ypart) - \RelevanceReg{\w}(\ypart|\x_t)\right)^2 \\
    &= \L^{\relx}(\w)
\end{align*}
Line 2 formulates the expectation in terms of the marginal exposure probabilities $\Prob(\expx_t(\ypart)|\y_t,\x_t)$, which decomposes the expectation as the objective is additive in $\ypart$. Note that $\Prob(\expx_t(\ypart)=1|\y_t,\x_t)$ is therefore equal to $\prop_t(\ypart)$ under our exposure model. Line 3 substitutes $\click_t(\ypart)=\expx_t(\ypart) \relx_t(\ypart)$ and simplifies the expression, since $\expx_t(\ypart) \relx_t(\ypart) = 0$ whenever the user is not exposed to an item. Note that the propensities $\prop_t(\y)$ for the exposed items now cancel, as long as they are bounded away from zero -- meaning that all items have some probability of being found by the user. In case users do not naturally explore low enough in the ranking, active interventions can be used to stochastically promote items in order to ensure non-zero examination propensities (e.g. \cite{hofmann2013balancing}). 
Note that unbiasedness holds for any sequence of $(\x_1,\relx_1,\y_1) ..., (\x_T,\relx_T,\y_T)$, no matter how complex the dependencies between the rankings $\y_t$ are.

Beyond this proof of unbiasedness, it is possible to use standard concentration inequalities to show that $\L^{\click}(\w)$ converges to $\L^{\relx}(\w)$ as the size $\tau$ of the training sequence increases. 
Thus, under standard conditions on the capacity for uniform convergence, it is possible to show convergence of the minimizer of $\L^{\click}(\w)$ to the least-squares regressor as the size $\tau$ of the training sequence increases.  We will use this regression objective to learn neural-network rankers in Section~\ref{sec:expjester}.

\subsection{Estimating Average Relevances}

The conditional relevances $\Relevance(\ypart|\x)$ are used in the ranking policies from Equation~\eqref{eq:sortrank}. But when defining merit in Equation~\eqref{eq:merit} for the fairness disparities, the average relevance $\Relevance(\ypart)$ is needed. Furthermore, $\Relevance(\ypart)$ serves as the ranking criterion for global rankings in Equation~\eqref{eq:sortrankglob}. While we could marginalize $\Relevance(\ypart|\x)$ over $\Prob(\x)$ to derive $\Relevance(\ypart)$, we argue that the following is a more direct way to get an unbiased estimate.
\begin{eqnarray}
    \RelevanceIPS(\ypart) & = & \frac{1}{\tau} \sum_{t=1}^\tau  \frac{\click_t(\ypart)}{\prop_t(\ypart)} \label{eq:relips}.
\end{eqnarray}
The following shows that this estimator is unbiased as long as the propensities are bounded away from zero.
\begin{align*}
    \Exp_{\expx}\left[ \RelevanceIPS(\ypart) \right] &= \frac{1}{\tau} \sum_{t=1}^\tau  \sum_{\expx_t(\ypart)} \frac{\expx_t(\ypart) \relx_t(\ypart)}{\prop_t(\ypart)} \Prob(\expx_t(\ypart)|\y_t,\x_t)\\
    &= \frac{1}{\tau} \sum_{t=1}^\tau \frac{ \relx_t(\ypart)}{\prop_t(\ypart)} \prop_t(\ypart)\\
    &= \frac{1}{\tau} \sum_{t=1}^\tau  \relx_t(\ypart)\\
    &= \Relevance(\ypart)
\end{align*}
In the following experiments, we will use this estimator whenever a direct estimate of $\Relevance(\ypart)$ is needed for the fairness disparities or as a global ranking criterion.

\section{Dynamically Controlling Fairness} \label{sec:faircontrol}

Given the formalization of the dynamic LTR problem, our definition of fairness, and our derivation of estimators for all relevant parameters, we are now in the position to tackle the problem of ranking while enforcing the fairness conditions. We view this as a control problem since we need to be robust to the uncertainty in the estimates $\RelevanceEst(\ypart|\x)$ and $\RelevanceEst(\ypart)$ at the beginning of the learning process. Specifically, we propose a controller that is able to make up for the initial uncertainty as these estimates converge during the learning process. 

Following our pairwise definitions of amortized fairness from Section~\ref{sec:fairness}, we quantify by how much fairness between all classes is violated using the following overall disparity metric.
\begin{equation}
    \Dibar_{\tau}= \frac{2}{m(m-1)} \sum_{i=0}^{m}\sum_{j=i+1}^m \left| \Di_{\tau}(G_i,G_j) \right|  \label{eq:disparityEmean}
\end{equation}
This metric can be instantiated with the disparity $\DiE_{\tau}(G_i,G_j)$ from Equation~\eqref{eq:fairexposure} for exposure-based fairness, or $\DiI_{\tau}(G_i,G_j)$ from Equation~\eqref{eq:fairimpact} for impact-based fairness. Since optimal fairness is achieved for $\Dibar_{\tau}=0$, we seek to minimize $\Dibar_{\tau}$.

To this end, we now derive a method we call {\it \conname}, which takes the form of a Proportional Controller (a.k.a. P-Controller) \cite{bequette2003process}. A P-controller is a widely used control-loop mechanism that applies feedback through a correction term that is proportional to the error. In our application, the error corresponds to the violation of our amortized fairness disparity from Equations~\eqref{eq:fairexposure} and \eqref{eq:fairimpact}. 
Specifically, for any set of disjoint groups $\mathcal{G} = \{G_1,\dots,G_m\}$, the error term of the controller for any item $\ypart$ is defined as
\[\forall G\in\mathcal{G} \: \forall \ypart \in G: \err_\tau(\ypart) =  (\tau-1) \cdot \max_{G_i}\bigg(  \DiEst_{\tau-1}(G_i,G) \bigg). \]
The error term $\err_\tau(G)$ is zero for the group that already has the maximum exposure/impact w.r.t. its merit. For items in the other groups, the error term grows with increasing disparity.

Note that the disparity $\DiEst_{\tau-1}(G_i,G)$ in the error term uses the estimated  $\MeritEst(G)$ from Equation~\eqref{eq:relips}, which converges to $\Merit(G)$ as the sample size $\tau$ increases. To avoid division by zero, $\MeritEst(G)$ can be set to some minimum constant.

We are now in a position to state the \conname~ranking policy as
\begin{eqnarray}
    \textrm{\conname:} \qquad \y_\tau = \argsort_{\ypart \in \Y}\:\left( \RelevanceEst(\ypart|\x) + \lambda\; \err_\tau(\ypart)\right). \label{eq:p-controller}
\end{eqnarray}
When the exposure-based disparity $\DiEEst_{\tau-1}(G_i,G)$ is used in the error term, we refer to this policy as \connameE. If the impact-based disparity $\DiIEst_{\tau-1}(G_i,G)$ is used, we refer to it as \connameI. 

Like the policies in Section~\ref{sec:sortrank}, \conname~is a sort-based policy. However, the sorting criterion is a combination of relevance $\RelevanceEst(\ypart|\x)$ and an error term representing the fairness violation. 
The idea behind \conname~is that the error term pushes the items from the underexposed groups upwards in the ranking. The parameter $\lambda$ can be chosen to be any positive constant. While any choice of $\lambda$ leads to asymptotic convergence as shown by the theorem below for exposure fairness, a suitable choice of $\lambda$ can have influence on the finite-sample behavior of \conname: a higher $\lambda$ can lead to an oscillating behavior, while a smaller $\lambda$ makes the convergence smoother but slower. We explore the role of $\lambda$ in the experiments, but find that keeping it fixed at $\lambda=0.01$ works well across all of our experiments. Another key quality of \conname~is that it is agnostic to the choice of error metric, and we conjecture that it can easily be adapted to other types of fairness disparities. Furthermore, it is easy to implement and it is very efficient, making it well suited for practical applications.

To illustrate the theoretical properties of \conname, we now analyze its convergence for the case of exposure-based fairness. To disentangle the convergence of the estimator for $\MeritEst(G)$ from the convergence of \conname, consider a time point $\tau_0$ where $\MeritEst(G)$ is already close to $\Merit(G)$ for all $G \in \mathcal{G}$. We can thus focus on the question whether \conname~can drive $\DiEbar_{\tau}$ to zero starting from any unfairness that may have persisted at time $\tau_0$. To make this problem well-posed, we need to assume that exposure is not available in overabundance, otherwise it may be unavoidable to give some groups more exposure than they deserve even if they are put at the bottom of the ranking. A sufficient condition for excluding this case is to only consider problems for which the following is true: for all pairs of groups $G_i, G_j$, if $G_i$ is ranked entirely above $G_j$ at any time point $t$, then 
\begin{equation}
    \frac{\Exposure_t(G_i)}{\MeritEst(G_i)} \geq \frac{ \Exposure_t(G_j)}{\MeritEst(G_j)}. \label{eq:non-overabundance}
\end{equation}
Intuitively, the condition states that ranking $G_i$ ahead of $G_j$ reduces the disparity if $G_i$ has been underexposed in the past.
We can now state the following theorem.

\begin{theorem}\label{theorem:1}
For any set of disjoint groups $\mathcal{G} = \{G_1,\dots,G_m\}$ with any fixed target merits $\MeritEst(G_i)>0$ that fulfill \eqref{eq:non-overabundance}, any relevance model $\RelevanceEst(\ypart|\x) \in [0,1]$, any exposure model $\prop_t(\ypart)$ with $0 \leq \prop_t(\ypart) \leq \prop_{\max}$, and any value $\lambda>0$, running \connameE~from time $\tau_0$ will always ensure that the overall disparity $\Dibar^E_\tau$ with respect to the target merits converges to zero at a rate of $\mathcal{O}\left(\frac{1}{\tau}\right)$, no matter how unfair the exposures $\frac{1}{\tau_0}\sum_{t=1}^{\tau_0} \Exposure_t(G_j)$ up to $\tau_0$ have been.
\end{theorem}

The proof of the theorem is included in Appendix~\ref{sec:theoremproof}. Note that this theorem holds for any time point $\tau_0$, even if the estimated merits change substantially up to $\tau_0$. So, once the estimated merits have converged to the true merits, \connameE~will ensure that the amortized disparity $\DiEbar_{\tau}$ converges to zero as well.

\section{Empirical Evaluation} \label{sec:exp}

In addition to the theoretical justification of our approach, we also conducted an empirical evaluation\footnote{The implementation is available at \url{https://github.com/MarcoMorik/Dynamic-Fairness}.}.
We first present experiments on a semi-synthetic news dataset to investigate different aspects of the proposed methods under controlled conditions. After that we evaluate the methods on real-world movie preference data for external validity.

\subsection{Robustness Analysis on News Data} \label{sec:expsim}

To be able to evaluate the methods in a variety of specifically designed test settings, we created the following simulation environment from articles in the Ad Fontes Media Bias dataset\footnote{https://www.adfontesmedia.com/interactive-media-bias-chart/}. It simulates a dynamic ranking problem on a set of news articles belonging to two groups $G_\lft$ and $G_\rght$ (e.g. left-leaning and right-leaning news articles).

In each trial, we sample a set of 30 news articles $\Y$. For each article, the dataset contains a polarity value $\polarity^\ypart$ that we rescale to the interval between -1 and 1, while the user polarities are simulated. Each user has a polarity that is drawn from a mixture of two normal distributions clipped to $[-1, 1]$
\begin{equation}
    \polarity^{u_t} \sim \text{clip}_{[-1,1]}\big( p_{neg} \mathcal{N}(-0.5,0.2) + (1-p_{neg}) \mathcal{N}(0.5,0.2)\big) \label{eq:polarity_user}
\end{equation}
where $p_{neg}$ is the probability of the user to be left-leaning (mean=$-0.5$). We use $p_{neg}=0.5$ unless specified. In addition, each user has an openness parameter $o^{u_t} \sim \mathcal{U}(0.05,0.55)$, indicating on the breadth of interest outside their polarity.
Based on the polarities of the user $u_t$ and the item $\ypart$, the true relevance is drawn from the Bernoulli distribution 
\[\relx_t(\ypart) \sim \textrm{Bernoulli}\left[p=exp\left(\frac{-(\polarity^{u_t}-\polarity^{\ypart})^2}{2(o^{u_t})^2}\right)\right].\]

As the model of user behavior, we use the Position-based click model (PBM \cite{chuklin2015click}), where the marginal probability that user $u_t$ examines an article only depends only on its position. We choose an exposure drop-off analogous to the gain function in DCG as 
\begin{equation}
\prop_t(\ypart) = \frac{1}{\log_2(\rank(\ypart|\y_t)+1)}. \label{eq:propensitycurve}
\end{equation}

The remainder of the simulation follows the dynamic ranking setup. At each time step $t$ a user $u_t$ arrives to the system, the algorithm presents an unpersonalized ranking $\y_t$, and the user provides feedback $\click_t$ according to $\prop_t$ and $\relx_t$. The algorithm only observes $\click_t$ and not $\relx_t$.

To investigate group-fairness, we group the items according to their polarity, where items with a polarity $\polarity^{\ypart} \in [-1,0)$ belong to the \emph{left-leaning} group $G_{\lft}$ and items with a polarity $\polarity^{\ypart} \in [0,1]$ belong to the \emph{right-leaning} group $G_{\rght}$. 

We measure ranking quality by the average cumulative NDCG $\frac{1}{\tau}\sum_{t=1}^\tau \Loss^{DCG}(\y_t|\relx_t)$ over all the users up to time $\tau$. We measure Exposure Unfairness via $\DiEbar_{\tau}$ and Impact Unfairness via $\DiIbar_{\tau}$ as defined in Equation~\eqref{eq:disparityEmean}.

In all news experiments, we learn a global ranking and compare the following methods. 
\begin{description}
    \item[\naive:] Rank by the sum of the observed feedback $\click_t$. 
    \item[\ultrglob:] Dynamic LTR by sorting via the unbiased estimates $\RelevanceIPS(\ypart)$ from Eq.~\eqref{eq:relips}.
    \item[\connameI:] Fairness controller from Eq.~\eqref{eq:p-controller} for impact fairness.
\end{description}

\subsubsection{Can \conname\ reduce unfairness while maintaining good ranking quality?}

\begin{figure}[h!]
    \centering
    \includegraphics[width=\linewidth]{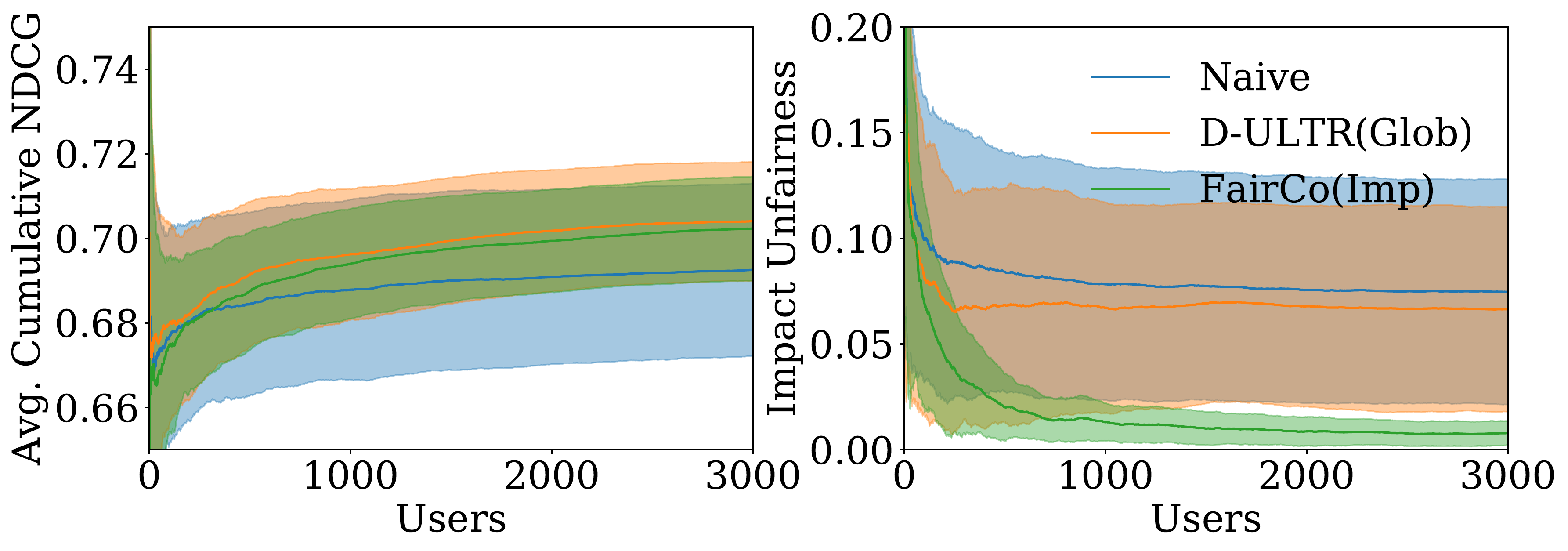}
    \caption{Convergence of NDCG (left) and Unfairness (right) as the number of users increases.
    (100 trials) } 
    \label{fig:syntheticImpact}
\end{figure}
This is the key question in evaluating \conname, and Figure~\ref{fig:syntheticImpact} shows how NDCG and Unfairness converge for \naive, \ultrglob, and \connameI. The plots show that \naive\ achieves the lowest NDCG and that its unfairness remains high as the number of user interactions increases. \ultrglob\ achieve the best NDCG, as predicted by the theory, but its unfairness is only marginally better than that of \naive. Only \conname\ manages to substantially reduce unfairness, and this comes only at a small decrease in NDCG compared to \ultrglob.

The following questions will provide further insight into these results, evaluating the components of the \conname\ and exploring its robustness.

\subsubsection{Do the unbiased estimates converge to the true relevances?}

\begin{figure}[h]
    \centering
    \includegraphics[width=\linewidth]{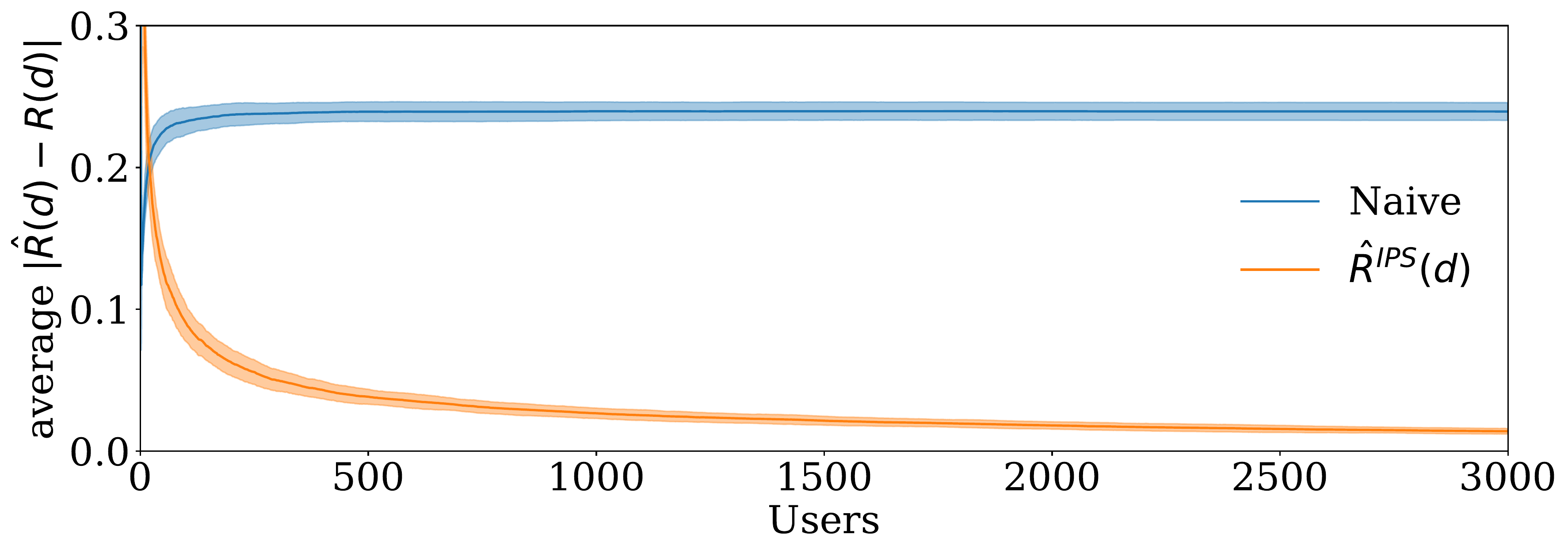}
    \caption{Error of relevance estimators as the number of users increases ($|\Y|=30$, 10 trials)}
    \label{fig:relevanceConvergence}
\end{figure}

The first component of \conname\ we evaluate is the unbiased IPS estimator $\RelevanceIPS(\ypart)$ from Equation~\eqref{eq:relips}. %
Figure 1 shows the absolute difference between the estimated global relevance and true global relevance for $\RelevanceIPS(\ypart)$ and the estimator used in the \naive. While the error for \naive~stagnates at around 0.25, the estimation error of $\RelevanceIPS(\ypart)$ approaches zero as the number of users increases. This verifies that IPS eliminates the effect of position bias and learns accurate estimates of the true expected relevance for each news article so that we can use them for the fairness and ranking criteria. %

\subsubsection{Does \conname\ overcome the rich-get-richer dynamic?}

\begin{figure}[h]
    \centering
    \includegraphics[width=\linewidth]{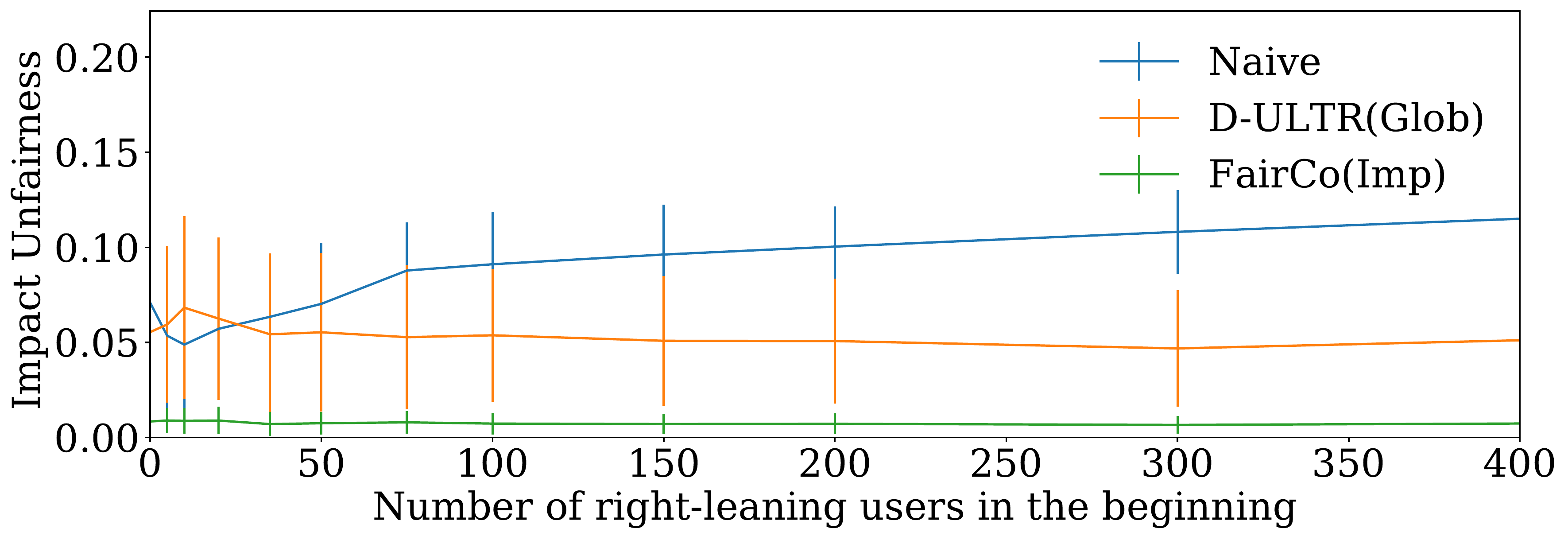}
    \caption{The effect of a block of right-leaning users on the Unfairness of Impact. (50 trials, 3000 users)}
    \label{fig:richgetricher}
\end{figure}

The illustrating example in Section~\ref{sec:motivation} argues that naively ranking items is highly sensitive to the initial conditions (e.g. which items get the first clicks), leading to a rich-get-richer dynamic. We now test whether \conname\ overcomes this problem. In particular, we adversarially modify the user distribution so that the first $x$ users are right-leaning ($p_{neg}=0$), followed by $x$ left-leaning users ($p_{neg}=1$), before we continue with a balanced user distribution ($p_{neg}=0.5$). Figure~\ref{fig:richgetricher} shows the unfairness after 3000 user interactions. As expected, \naive~is the most sensitive to the head-start that the right-leaning articles are getting. \ultrglob~fares better and its unfairness remains constant (but high) independent of the initial user distribution since the unbiased estimator $\RelevanceIPS(\ypart)$ corrects for the presentation bias so that the estimates still converge to the true relevance. \conname\ inherits this robustness to initial conditions since it uses the same $\RelevanceIPS(\ypart)$ estimator, and its active control for unfairness makes it the only method that achieves low unfairness across the whole range.

\subsubsection{How effective is the \conname\ compared to a more expensive Linear-Programming Baseline?}
\label{sec:fairness-cont-exp}
\begin{figure}[h!]
    \centering
    \includegraphics[width=\linewidth]{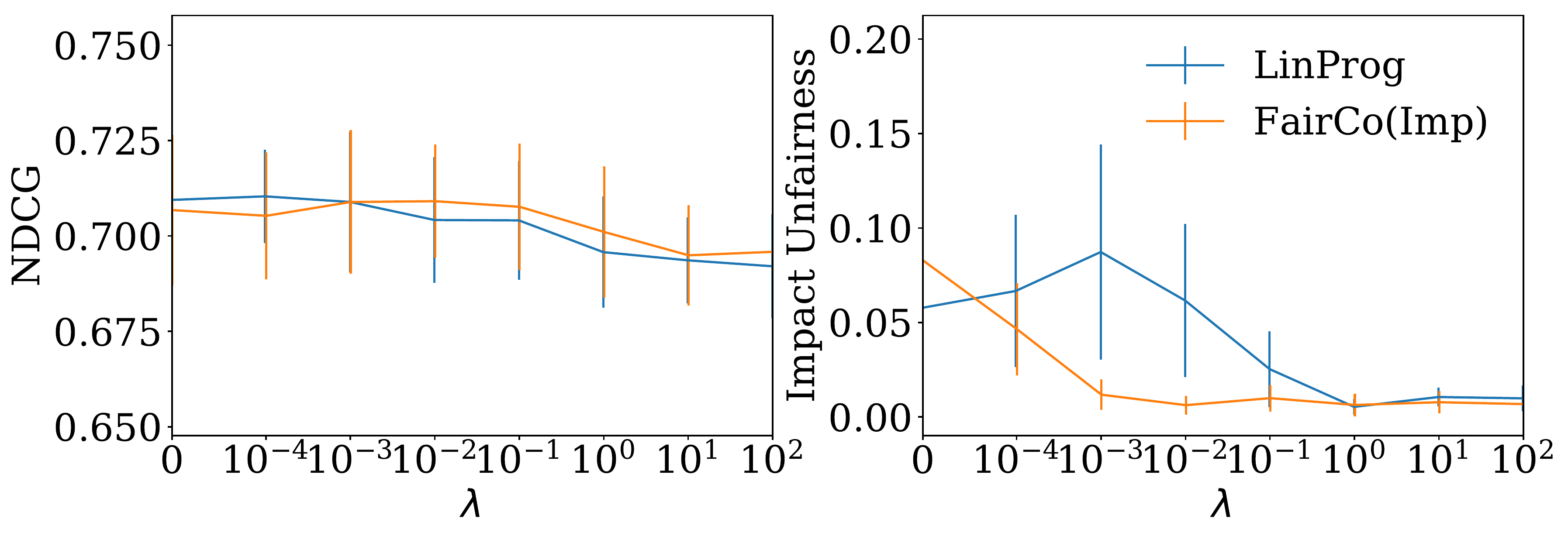}
    \caption{Comparing the LP Baseline and the P-Controller in terms of NDCG (left) and Unfairness (right) for different values of $\lambda$. (15 trials, 3000 users) }
    \label{fig:LPvsController}
\end{figure}
As a baseline, we adapt the linear programming method from \cite{singh2018fairness} to the dynamic LTR setting to minimize the amortized fairness disparities that we consider in this work. The method uses the current relevance and disparity estimates to solve a linear programming problem whose solution is a stochastic ranking policy that satisfies the fairness constraints in expectation at each $\tau$. The details of this method are described in Appendix~\ref{sec:fairlp}.
Figure \ref{fig:LPvsController} shows NDCG and Impact Unfairness after 3000 users averaged over 15 trials for both LinProg and \conname\ for different values of their hyperparameter $\lambda$. For $\lambda=0$, both methods reduce to \ultrglob\ and we can see that their solutions are unfair. As $\lambda$ increases, both methods start enforcing fairness at the expense of NDCG. In these and other experiments, we found no evidence that the LinProg baseline is superior to \conname. However, LinProg is substantially more expensive to compute, which makes \conname\ preferable in practice. %

\subsubsection{Is \conname\ effective for different group sizes?}

\begin{figure}[h!]
    \centering
    \includegraphics[width=\linewidth]{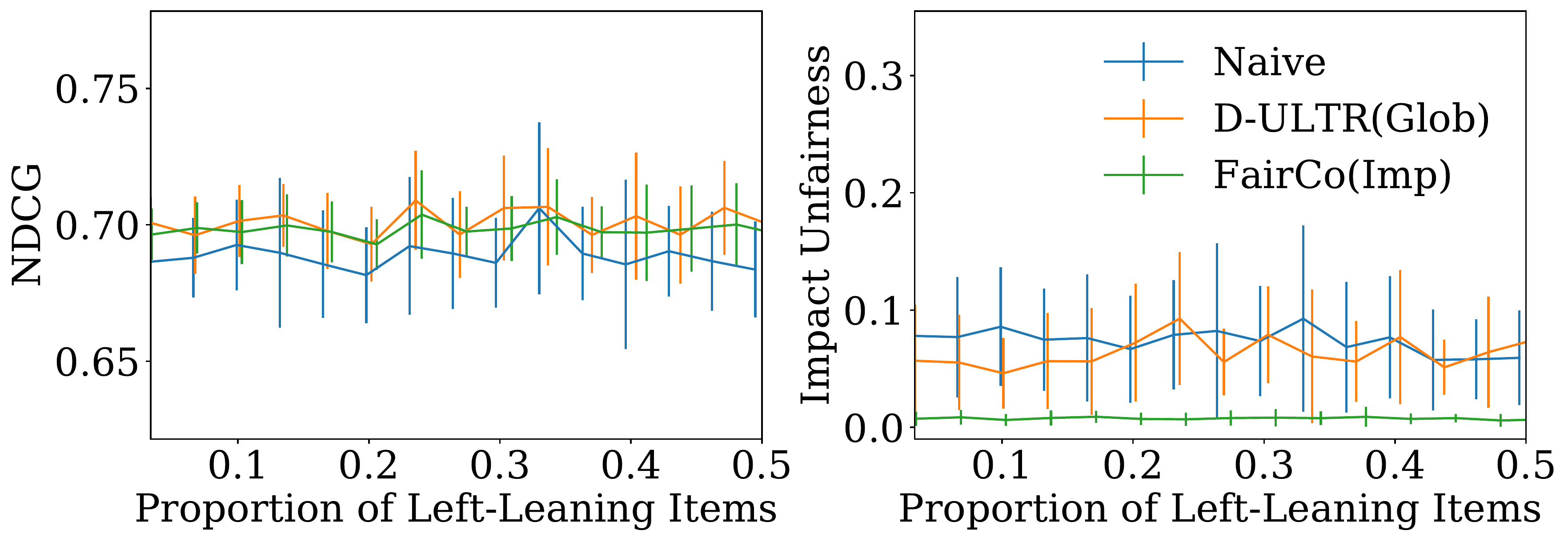}
    \caption{NDCG (left) and Unfairness (right) for varying proportion of $G_{\lft}$ (20 trials, 3000 users)}
    \label{fig:groupratio}
\end{figure}

In this experiment, we vary asymmetry of the polarity within the set of 30 news articles, ranging from $G_{\lft}=1$ to $G_{\lft}=15$ news articles. For each group size, we run 20 trials for 3000 users each. Figure \ref{fig:groupratio} shows that regardless of the group ratio, \conname\ reduces unfairness for the whole range while maintaining NDCG. This is in contrast to \naive~and \ultrglob, which suffer from high unfairness. 

\subsubsection{Is \conname\ effective for different user distributions?}

\begin{figure}[h!]
    \centering
    \includegraphics[width=\linewidth]{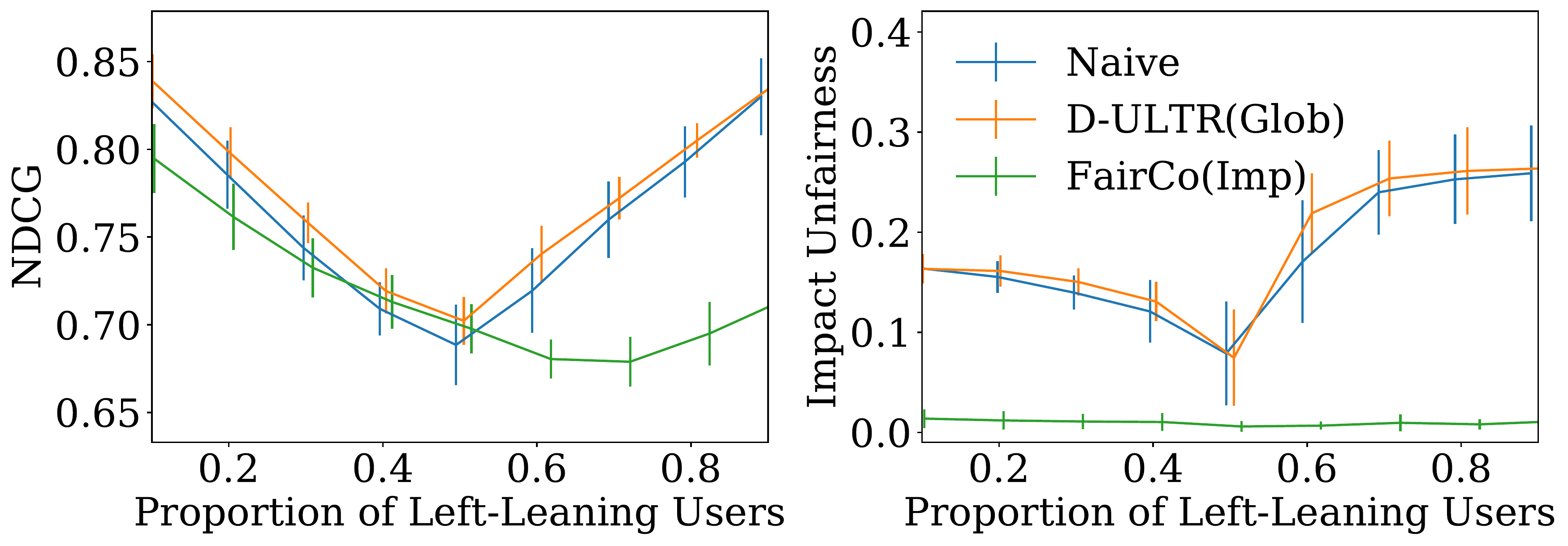}
    \caption{NDCG (left) and Unfairness (right) with varying user distributions. (20 trials, 3000 users)}
    \label{fig:userdistribution}
\end{figure}
Finally, to examine the robustness to varying user distributions, we control the polarity distribution of the users by varying $p_{neg}$ in Equation~\eqref{eq:polarity_user}. We run 20 trials each on 3000 users. In Figure~\ref{fig:userdistribution}, observe that \naive\ and \ultrglob\ suffer from high unfairness when there is a large imbalance between the minority and the majority group, while \conname\ is able to control the unfairness in all settings. %

\subsection{Evaluation on Real-World Preference Data} \label{sec:expjester}

To evaluate our method on a real-world preference data, we adopt the ML-20M dataset \cite{harper2015movielens}. 
We select the five production companies with the most movies in the dataset --- \textit{MGM, Warner Bros, Paramount, 20th Century Fox, Columbia}. These production companies form the groups for which we aim to ensure fairness of exposure. To exclude movies with only a few ratings and have a diverse user population, from the set of $300$ most rated movies by these production companies, we select $100$ movies with the highest standard deviation in the rating across users. For the users, we select $10^4$ users who have rated the most number of the chosen movies.
This leaves us with a partially filled ratings matrix with $10^4$ users and $100$ movies. To avoid missing data for the ease of evaluation, we use an off-the-shelf matrix factorization algorithm\footnote{Surprise library (http://surpriselib.com/) for SVD with \texttt{biased=False} and \texttt{D=50}} to fill in the missing entries. We then normalize the ratings to $[0,1]$ by apply a Sigmoid function centered at rating $b=3$ with slope $a=10$. These serve as relevance probabilities where higher star ratings correspond to a higher likelihood of positive feedback. Finally, for each trial we obtain a binary relevance matrix by drawing a Bernoulli sample for each user and movie pair with these probabilities. We use the user embeddings from the matrix factorization model as the user features $\x_t$. 

In the following experiments we use \conname\ to learn a sequence of ranking policies $\pol_t(\x)$ that are personalized based on $\x$. The goal is to maximize NDCG while providing fairness of exposure to the production companies. User interactions are simulated analogously to the previous experiments. At each time step $t$, we sample a user $\x_t$ and the ranking algorithm presents a ranking of the 100 movies. The user follows the position-based model from Equation~\eqref{eq:propensitycurve} and reveal $\click_t$ accordingly.

For the conditional relevance model $\RelevanceRegStar(\ypart|\x)$ used by \conname\ and \ultr, we use a one hidden-layer neural network that consists of $D=50$ input nodes fully connected to 64 nodes in the hidden layer with ReLU activation, which is connected to 100 output nodes with Sigmoid to output the predicted probability of relevance of each movie. Since training this network with less than 100 observations is unreliable, we use the global ranker \ultrglob\ for the first 100 users. We then train the network at $\tau=100$ users, and then update the network after every 10 users on all previously collected feedback i.e. $\click_1, ..., \click_{\tau}$ using the unbiased regression objective, $\L^{\click}(\w)$, from Eq.~\eqref{eq:newloss2} with the Adam optimizer \cite{kingma2014adam}.

\subsubsection{Does personalization via unbiased objective improve NDCG?}

\begin{figure}[h!]
    \vspace{-0.2cm}
    \centering
    \includegraphics[width=\linewidth]{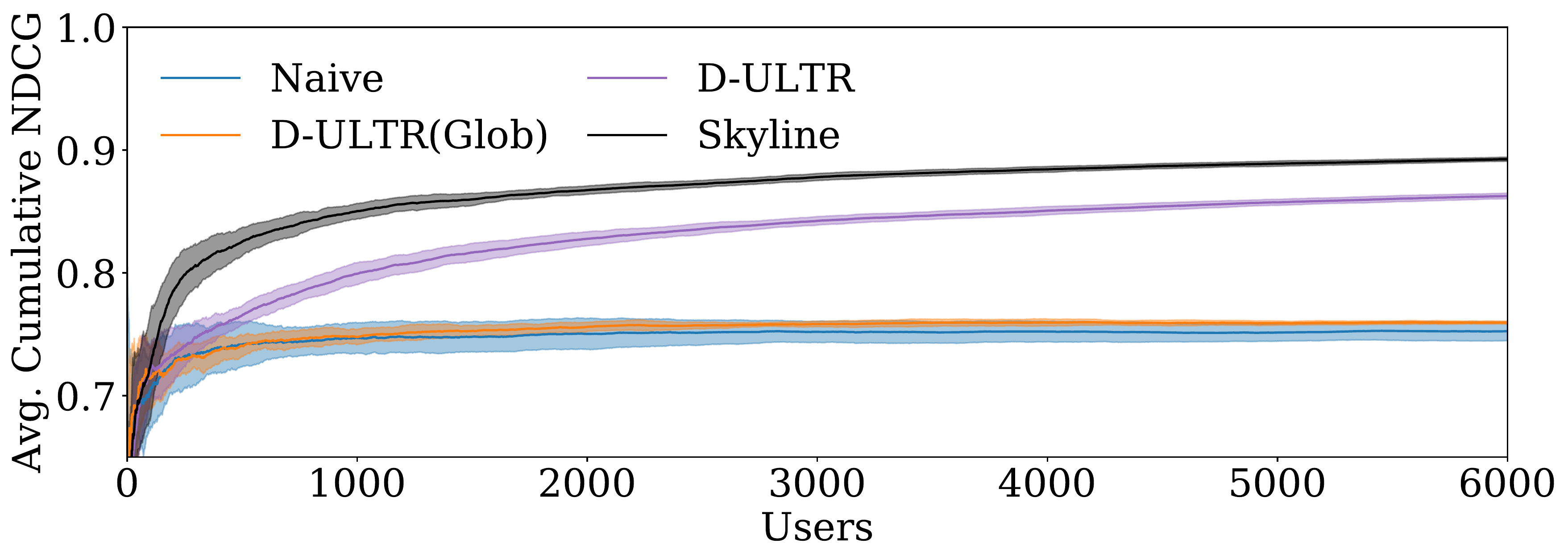}
    \caption{Comparing the NDCG of personalized and non-personalized rankings on the Movie data. (10 trials)}
    \vspace{-0.1cm}
    \label{fig:PersNDCG}
\end{figure}

We first evaluate whether training a personalized model using the de-biased $\RelevanceRegStar(\ypart|\x)$ regression estimator improves ranking performance over a non-personalized model. 
Figure \ref{fig:PersNDCG} shows that ranking by $\RelevanceRegStar(\ypart|\x)$ (i.e. \ultr) provides substantially higher NDCG than the unbiased global ranking \ultrglob\ and the \naive\ ranking. 
To get an upper bound on the performance of the personalization models, we also train a \skylineultr\ model using the (in practice unobserved) true relevances $\relx_t$ with the least-squares objective from Eq.~\eqref{eq:regrel}. 
Even though the unbiased regression estimator $\RelevanceRegStar(\ypart|\x)$ only has access to the partial feedback $\click_t$, it tracks the performance of \skylineultr. As predicted by the theory, they appear to converge asymptotically.

\subsubsection{Can \conname\ reduce unfairness?}
\begin{figure}[t!]
    \centering
    \includegraphics[width=\linewidth]{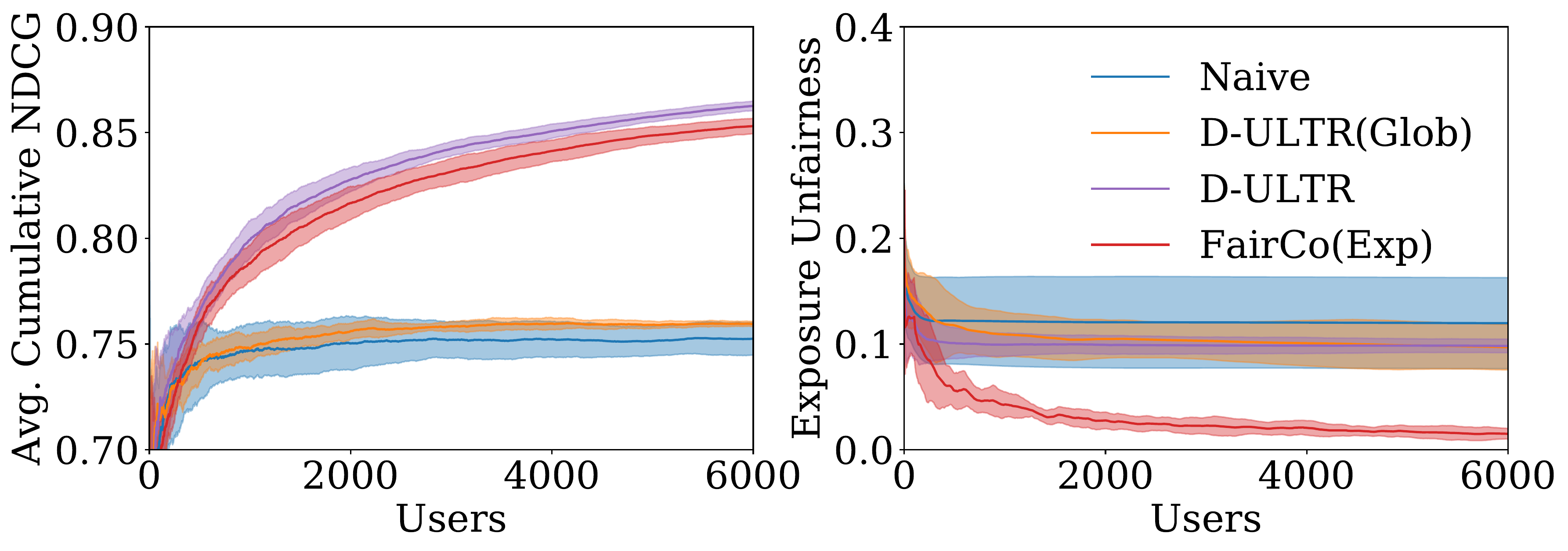}
    \caption{NDCG (left) and Exposure Unfairness (right) on the Movie data as the number of user interactions increases. (10 trials)}
    \label{fig:RealExposure}
\end{figure}
\begin{figure}
    \vspace{-0.15cm}
    \centering
    \includegraphics[width=\linewidth]{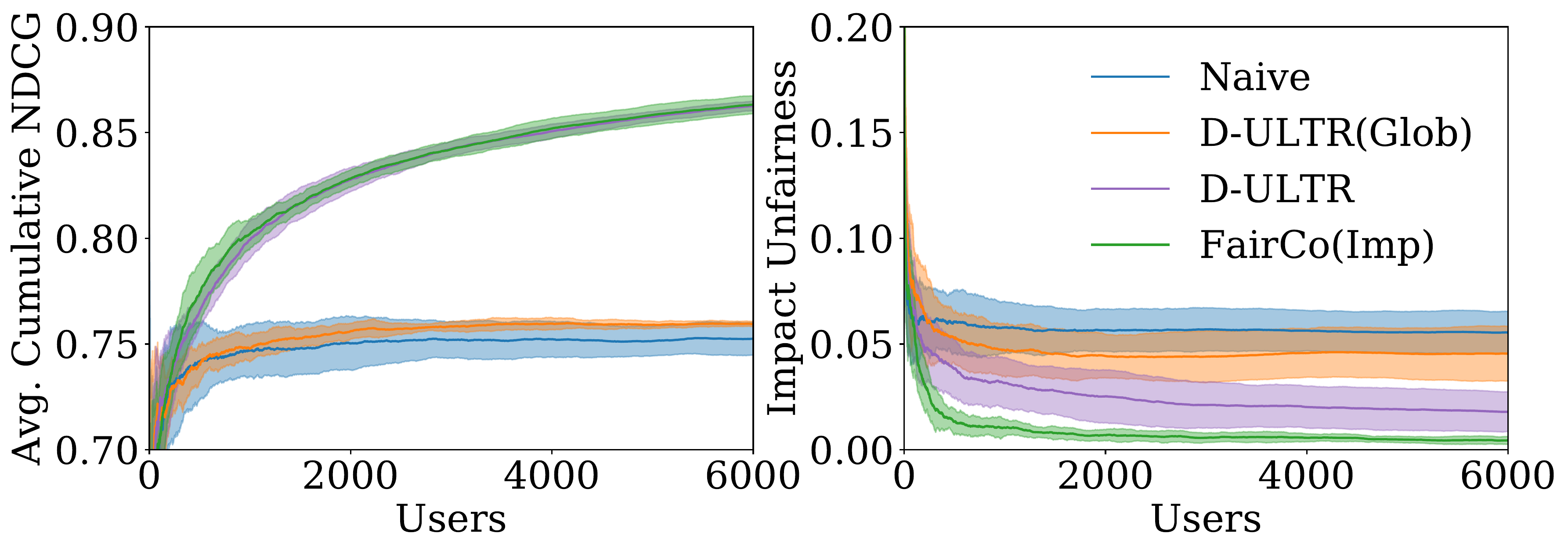}
    \caption{NDCG (left) and Impact Unfairness (right) on the Movie data as the number of user interactions increases. (10 trials)}
    \label{fig:RealImpact}
    \vspace{-0.1cm}
\end{figure}

Figure~\ref{fig:RealExposure} shows that \connameE\ can effectively control Exposure Unfairness, unlike the other methods that do not actively consider fairness. Similarly, Figure~\ref{fig:RealImpact} shows that \connameI\ is effective at controlling Impact Unfairness. As expected, the improvement in fairness comes at a reduction in NDCG, but this reduction is small. %

\subsubsection{How different are exposure and impact fairness?}

\begin{figure}[h]
\vspace{-0.15cm}
    \centering
    \includegraphics[width=\linewidth]{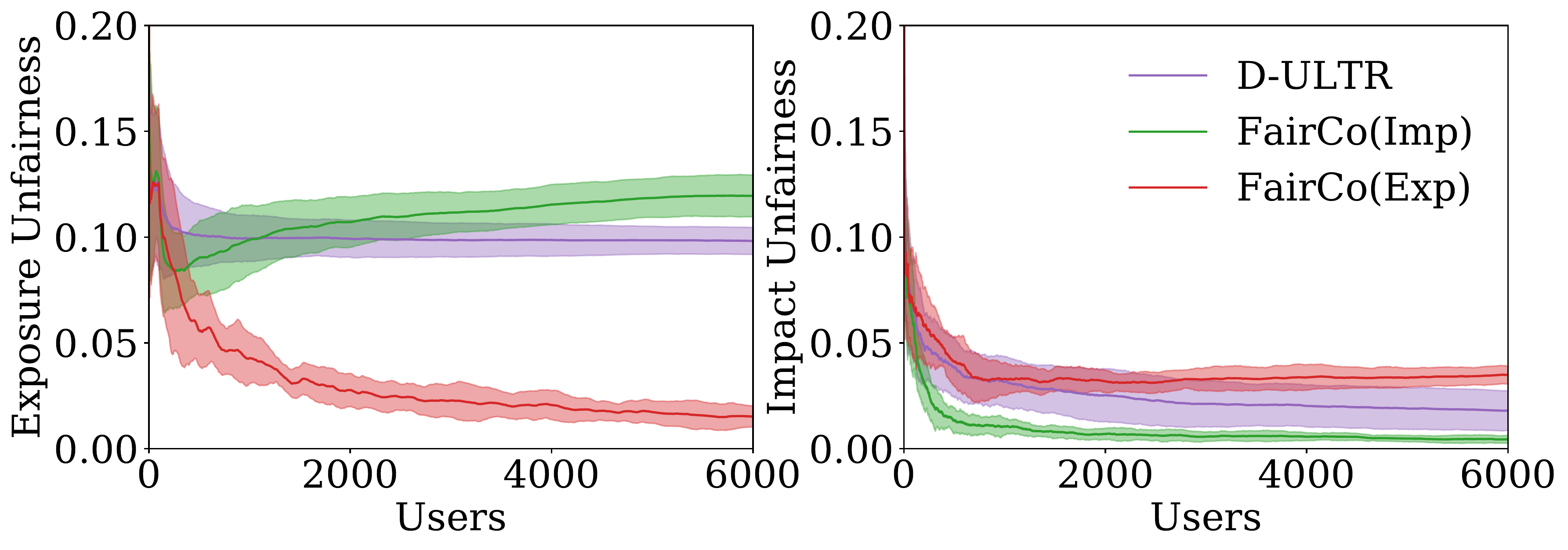}
    \caption{Unfairness of Exposure (left) and Unfairness of Impact (right) for the personalized controller optimized for either Exposure or Impact.  
     (10 trials)}
    \label{fig:ImpactvsExposure}
\vspace{-0.1cm}
\end{figure}

Figure \ref{fig:ImpactvsExposure} evaluates how an algorithm that optimizes Exposure Fairness performs in terms of Impact Fairness and vice versa. The plots show that the two criteria achieve different goals and that they are substantially different. In fact, optimizing for fairness in impact can even increase the unfairness in exposure, illustrating that the choice of criterion needs to be grounded in the requirements of the application.

\section{Conclusions}

We identify how biased feedback and uncontrolled exposure allocation can lead to unfairness and undesirable behavior in dynamic LTR. To address this problem, we propose \conname, which is able to adaptively enforce amortized merit-based fairness constraints even though their underlying relevances are still being learned. The algorithm is robust to presentation bias and thus does not exhibit rich-get-richer dynamics. Finally, \conname\ is easy to implement and computationally efficient, which makes it well suited for practical applications.

\begin{acks}
This research was supported in part by NSF Awards IIS-1901168 and a gift from Workday. All content represents the opinion of the authors, which is not necessarily shared or endorsed by their respective employers and/or sponsors.
\end{acks}

\bibliographystyle{ACM-Reference-Format}
\balance
\bibliography{refs,joachims}

\input{appendix}

\end{document}

%% file: macros.tex
\usepackage{amsmath}
\usepackage{amssymb}
\usepackage{dsfont}
\usepackage{booktabs}
\usepackage{bbm}
\usepackage[para]{footmisc}
\usepackage{enumitem}
\DeclareMathOperator{\Exp}{\mathbb{E}} %
\DeclareMathOperator{\relx}{\mathbf{r}}
\DeclareMathOperator{\click}{\mathbf{c}}
\DeclareMathOperator{\expx}{\mathbf{e}}
\DeclareMathOperator{\prop}{\mathbf{p}}
\DeclareMathOperator{\rank}{rank}
\DeclareMathOperator{\dynrank}{{\mathcal A}}
\DeclareMathOperator{\argmax}{argmax}

\DeclareMathOperator{\Prob}{P}
\newcommand{\x}{\bm{x}}                  %
\newcommand{\Y}{{\mathcal{D}}}           %
\newcommand{\y}{\bm{\sigma}}                  %
\newcommand{\ypart}{d}                   %
\newcommand{\Loss}{U}               %
\newcommand{\Util}{U}                    %
\newcommand{\pol}{{\pi}}                 %
\newcommand{\w}{w}                       %

\newcommand{\Merit}{{Merit}}
\newcommand{\MeritEst}{{\hat{Merit}}}
\newcommand{\Relevance}{{R}}
\newcommand{\RelevanceEst}{\hat{R}}
\newcommand{\RelevanceIPS}{{\hat{R}^{\scalebox{0.6}{IPS}}}} %
\newcommand{\RelevanceReg}[1]{{\hat{R}^{\scalebox{0.7}{#1}}}}  %
\newcommand{\RelevanceRegStar}{{\hat{R}^{\scalebox{0.6}{Reg}}}}  %
\newcommand{\Exposure}{{Exp}}
\newcommand{\Impact}{{Imp}}

\newcommand{\err}{\mathbf{err}}
\newcommand{\argsort}{\operatornamewithlimits{argsort}}
\newcommand{\pmat}{\mathbb{P}}

\newcommand{\Dibar}{\overline{D}}
\newcommand{\DiE}{D^E}
\newcommand{\DiEEst}{\hat{D}^E}
\newcommand{\DiEbar}{\overline{D}^E}
\newcommand{\DiI}{D^I}
\newcommand{\DiIEst}{\hat{D}^I}
\newcommand{\DiIbar}{\overline{D}^I}
\newcommand{\Di}{D}
\newcommand{\DiEst}{\hat{D}}
\newcommand{\polarity}{\rho}

\newcommand{\conname}{{FairCo}}      %
\newcommand{\connameI}{{FairCo(Imp)}} %
\newcommand{\connameE}{{FairCo(Exp)}} %
\newcommand{\ultr}{{D-ULTR}}           %
\newcommand{\ultrglob}{{D-ULTR(Glob)}} %
\newcommand{\naive}{{Naive}}         %
\newcommand{\skylineultr}{{Skyline}} %

\newcommand{\lft}{\text{left}}
\newcommand{\rght}{\text{right}}

\renewcommand{\L}{\mathcal{L}}

%% file: appendix.tex
\newpage

\appendix
\section{Linear Programming Baseline} \label{sec:fairlp} 

Here, we present a version of the fairness constraint defined in \citet{singh2018fairness} that explicitly computes an optimal ranking to present in each time step $\tau$ by solving a linear program (LP). In particular, we formulate an LP that explicitly maximizes the estimated DCG of the ranking $\y_\tau$ while minimizing the estimated cumulative fairness disparity $\DiI_{\tau}$ formulated in Equation \eqref{eq:impg}. This is used as a baseline to compare the P-Controller with.

To avoid an expensive search over the exponentially-sized space of rankings as in \cite{biega2018equity}, we exploit an alternative representation \cite{singh2018fairness} as a doubly-stochastic matrix $\pmat$ that is sufficient for representing $\y_t$. In this matrix, the entry $\pmat_{y,j}$ denotes the probability of placing item $y$ at position $j$. Both DCG as well as $\Impact_{\tau}(G_i)$ are linear functions of $\pmat$, which means that the optimum can be computed as the following linear program.
\begin{align}
\pmat^* &= \:\underset{\pmat,\xi \ge 0}{\argmax}\:\: \underbrace{\sum_{y} \RelevanceEst(y|\x) \sum_{j=1}^n  \frac{\pmat_{y,j}}{\log(1+j)}}_{\rm Utility} - \lambda \sum_{i,j} \xi_{ij} \nonumber \\ 
&\!\!\text{s.t. }\:\forall y,j: \sum_{i=1}^n \pmat_{y,i} = 1, \:\:\:  \sum_{y'} \pmat_{y',j} = 1, \:\:\:  0 \leq \pmat_{y,j} \leq 1 \nonumber \\
& \forall\; G_i,G_j\!:\left(\frac{\hat{\Impact}_\tau(G_i | \pmat_\tau)}{\MeritEst(G_i)}   -\frac{\hat{\Impact}_\tau(G_j| \pmat_\tau)}{\MeritEst(G_j)} \right) + D_{\tau-1}(G_i,G_j)\leq \xi_{ij} %
\end{align}
The parameter $\lambda$ controls trade-off between DCG of $\y_t$ and fairness. We explore this parameter empirically in Section~\ref{sec:expsim}.

It remains to define $\hat{\Impact}_\tau(G_j| \pmat_\tau)$. Assuming the PBM click model with $q(j)$ denoting the examination propensity of item $\ypart$ at position $j$, the estimated probability of a click is $\RelevanceEst(\ypart)\cdot q(j)$. So we can estimate the impact on the items in group $G$ for the rankings defined by $\pmat$ as
\[\hat{\Impact}_\tau(G|\pmat) = \frac{1}{|G|}\sum_{\ypart \in G}\RelevanceEst(\ypart|\x) \left( \sum_{j=1}^n \pmat_{y,j} \;q(j)\right)\]
We use the \texttt{scipy.optimize.linprog} LP solver to solve for the optimal $\pmat^*$, and then use a Birkhoff von-Neumann decomposition \cite{birkhoff1940lattice,singh2018fairness} to sample a deterministic ranking $\y_{\tau}$ to present to the user. This ranking is guaranteed to achieve the DCG and fairness optimized by $\pmat^*$ in expectation.

Note that the number of variables in the LP is $O(n^2 + |\mathcal{G}|^2)$, and even a polynomial-time LP solver incurs substantial computation cost when working with a large number of items in a practical dynamic ranking application.

\newpage 

\section{Convergence of \conname-Controller}
\label{sec:theoremproof}
In this section we will prove the convergence theorem of \conname~for exposure fairness. We conjecture that analogous proofs apply to other fairness criteria as well. 
To prove the main theorem, we will first set up the following lemmas.
\begin{lemma}
\label{lemma:1}
Under the conditions of the main theorem, for any value of $\lambda$ and any $\tau>\tau_0$: if $\DiE_{\tau-1}(G_i, G_j) > \frac{1}{(\tau-1) \lambda}$, then $$\tau \DiE_\tau(G_i, G_j) \leq (\tau-1)\DiE_{\tau-1}(G_i, G_j).$$ 
\end{lemma}
\begin{proof}
From the definition of $\DiE_\tau$ in Eq.~\eqref{eq:fairexposure} we know that for $\tau>\tau_0$, 
\[\tau \DiE_{\tau}(G_i, G_j) = (\tau-1)\DiE_{\tau-1}(G_i, G_j) + \left(\frac{\Exposure_{\tau}(G_i)}{\MeritEst(G_i)}-\frac{\Exposure_{\tau}(G_j)}{\MeritEst(G_j)}\right).\] 
Since $\DiE_{\tau-1}(G_i, G_j) > \frac{1}{(\tau-1) \lambda}$, we know that for all items in $G_j$ it holds that $\err_{\tau}(\ypart) > \frac{1}{\lambda}$. Hence, \conname~adds a correction term $\lambda \err_{\tau}(\ypart)$ to the $\RelevanceEst(\ypart)$ of all $d\in G_j$ that is greater than $\lambda\frac{1}{\lambda}=1$. Since $0 \leq \RelevanceEst(d)\leq 1$, the ranking is dominated by the correction term $\lambda\err_{\tau}(\ypart)$. This means that all $d\in G_j$ are ranked above all $d\in G_i$. Under the feasibility condition from Eq.\eqref{eq:non-overabundance}, this implies that $\left(\frac{\Exposure_{\tau}(G_i)}{\MeritEst(G_i)} \leq \frac{\Exposure_{\tau}(G_j)}{\MeritEst(G_j)}\right)$ and thus $\tau \DiE_{\tau}(G_i, G_j) \leq (\tau-1)\DiE_{\tau-1}(G_i, G_j)$.%
\end{proof}

\begin{lemma}
\label{lemma:2} Under the conditions of the main theorem, for any value of $\lambda>0$ there exists $\Delta \geq 0$ such that for any $G_i, G_j$ and $\tau > \tau_0$: if $\DiE_{\tau-1}(G_i, G_j) \leq \frac{1}{(\tau-1)\lambda}$, then $\tau \DiE_\tau(G_i, G_j) \leq \frac{1}{\lambda} + \Delta$. 
\end{lemma}
\begin{proof}
Using the definition the definition of $\DiE_\tau$ in Eq.~\eqref{eq:fairexposure}, we know that
\begin{align*}
    \tau \DiE_{\tau}(G_i, G_j) &= (\tau-1) \DiE_{\tau-1}(G_i, G_j) +  \left(\frac{\Exposure_\tau(G_i)}{\MeritEst(G_i)}-\frac{\Exposure_\tau(G_j)}{\MeritEst(G_j)}\right)\\
        &\leq \frac{1}{\lambda} + \left(\frac{\Exposure_\tau(G_i)}{\MeritEst(G_i)}-\frac{\Exposure_\tau(G_j)}{\MeritEst(G_j)}\right)\\
        & \leq \frac{1}{\lambda} + \Delta
\end{align*}
where $\Delta = \max_\sigma \max_{\substack{G, G'\\G \neq G'}} \left(\frac{\Exposure_\sigma(G)}{\MeritEst(G)}-\frac{\Exposure_\sigma(G')}{\MeritEst(G')}\right)$. Note that $\Delta$ is a constant independent of $\tau$ and refers to the ranking $\sigma$ for which two groups $G, G'$ have the maximum exposure difference (e.g. one is placed at the top of the ranking, and the other is placed at the bottom).
\end{proof}

Using these two lemmas, we conclude the following theorem:

\begin{theorem}\label{theorem:appendix}
For any set of disjoint groups $\mathcal{G} = \{G_1,\dots,G_m\}$ with any fixed target merits $\MeritEst(G_i)>0$ that fulfill \eqref{eq:non-overabundance}, any relevance model $\RelevanceEst(\ypart|\x) \in [0,1]$, any exposure model $\prop_t(\ypart)$ with $0 \leq \prop_t(\ypart) \leq \prop_{\max}$, and any value $\lambda>0$, running \connameE~from time $\tau_0$ will always ensure that the overall disparity $\Dibar^E_\tau$ with respect to the target merits converges to zero at a rate of $\mathcal{O}\left(\frac{1}{\tau}\right)$, no matter how unfair the exposures $\frac{1}{\tau_0}\sum_{t=1}^{\tau_0} \Exposure_t(G_j)$ up to $\tau_0$ have been.
\end{theorem}

\begin{proof}
To prove that $\DiEbar_\tau$ converges to zero at a rate of $\mathcal{O}\left(\frac{1}{\tau}\right)$, we will show that for all $\tau \geq \tau_0$, the following holds:
\[\DiEbar_\tau \leq \frac{1}{\tau}\frac{2}{m(m-1)}\sum_{i=1}^m\sum_{j=i+1}^m \max\left(\tau_0\left|\DiE_{\tau_0}(G_i, G_j)\right|, \frac{1}{\lambda} + \Delta\right)\]

The two terms in the max provide an upper bound on the disparity at time $\tau$ for any $G_i$ and $G_j$. To show this, we prove by induction that $\tau \DiE_{\tau}(G_i,G_j) \le \max\left(\tau_0\left|\DiE_{\tau_0}(G_i, G_j)\right|, \frac{1}{\lambda} + \Delta\right)$ for all $\tau \ge \tau_0$. At the start of the induction at $\tau=\tau_0$, the max directly upper bounds $\tau_0 \DiE_{\tau_0}(G_i,G_j)$. In the induction step from $\tau-1$ to $\tau$, if $(\tau-1)\DiE_{\tau-1}(G_i,G_j) > \frac{1}{\lambda}$, then Lemma~\ref{lemma:1} implies that $\tau\DiE_{\tau}(G_i,G_j) \le (\tau-1)\DiE_{\tau-1}(G_i,G_j) \le \max\left(\tau_0\left|\DiE_{\tau_0}(G_i, G_j)\right|, \frac{1}{\lambda} + \Delta\right)$.\\ If $(\tau-1)\DiE_{\tau-1}(G_i,G_j) \le \frac{1}{\lambda}$, then Lemma~\ref{lemma:2} implies that $\tau\DiE_{\tau}(G_i,G_j) \le \frac{1}{\lambda} + \Delta \le \max\left(\tau_0\left|\DiE_{\tau_0}(G_i, G_j)\right|, \frac{1}{\lambda} + \Delta\right)$ as well. This completes the induction, and we conclude that 
\[\DiE_{\tau}(G_i, G_j) \leq \frac{1}{\tau}\max\left(\tau_0|\DiE_{\tau_0}(G_i, G_j)|, \frac{1}{\lambda} + \Delta\right).\] Putting everything together, we get

\begin{align*}
    \DiEbar_{\tau} &= %
                     \frac{2}{m(m-1)} \sum_{i=0}^{m}\sum_{j=i+1}^m \left|\DiE_\tau(G_i, G_j)\right|\\
                    &\leq \frac{2}{m(m-1)} \sum_{i=0}^{m}\sum_{j=i+1}^m \left|\frac{1}{\tau}\max\left(\tau_0|\DiE_{\tau_0}(G_i, G_j)|, \frac{1}{\lambda} + \Delta\right)\right|\\
                    & \leq \frac{1}{\tau}\frac{2}{m(m-1)} \sum_{i=0}^{m}\sum_{j=i+1}^m \max\left(\tau_0\left|\DiE_{\tau_0}(G_i, G_j)\right|,\frac{1}{\lambda} + \Delta \right) \tag{since $\lambda, \Delta, \tau>0$}
\end{align*}
\end{proof}